\documentclass[twocolumn,longbibliography,nofootinbib, superscriptaddress,10pt,aps,prx]{revtex4-2}
\usepackage[dvips]{graphicx} 
\usepackage{amsfonts}
\usepackage{mathtools}
\usepackage{amssymb}
\usepackage{amscd}
\usepackage{amsmath}    
\usepackage{amsthm}
\usepackage{appendix}
\usepackage{bbm}
\usepackage{dsfont}
\usepackage{enumerate}
\usepackage{epsfig}
\usepackage[colorlinks, linkcolor=teal, anchorcolor=teal, citecolor=teal]{hyperref}
\usepackage{makecell}
\usepackage{xcolor}		
\usepackage{physics}
\usepackage[most]{tcolorbox}
\usepackage{tabularx}
\usepackage{MnSymbol}
\usepackage{tikz}
\usepackage{times}
\usepackage{tikzpeople}
\usetikzlibrary{arrows}
\usetikzlibrary{shapes,fadings,decorations}
\usetikzlibrary{decorations.pathmorphing,patterns}
\usetikzlibrary{calc}
\usetikzlibrary{positioning}

\newtheorem{theorem}{Theorem}
\newtheorem{fact}{Fact}
\newtheorem{lemma}{Lemma}
\newtheorem{corollary}{Corollary}

\newtheorem{observation}{Observation}

\newtcolorbox[auto counter]{mybox}[2][]{
	enhanced,
	colback=blue!5!white,
	colframe=blue!75!black,
	fonttitle=\bfseries,
	title=Box \thetcbcounter: #2,#1
}

\newcommand{\bE}{\mathbb{E}}

\newcommand{\TV}{\mathrm{TV}}

\allowdisplaybreaks

\begin{document}
\title{No universal purification in quantum mechanics}

\author{Zhenhuan Liu}
\thanks{ZL and ZD contributed equally to this work.\\\href{qubithuan@gmail.com}{qubithuan@gmail.com}\\ \href{mailto:du-zy23@mails.tsinghua.edu.cn}{du-zy23@mails.tsinghua.edu.cn}\\ \href{jense@zedat.fu-berlin.de}{jense@zedat.fu-berlin.de}\\ \href{cai.zhenyu.physics@gmail.com}{cai.zhenyu.physics@gmail.com}\\
\href{zwliu0@tsinghua.edu.cn}{zwliu0@tsinghua.edu.cn}}
\affiliation{Yau Mathematical Sciences Center, Tsinghua University, Beijing 100084, China}

\author{Zhenyu Du}
\thanks{ZL and ZD contributed equally to this work.\\\href{qubithuan@gmail.com}{qubithuan@gmail.com}\\ \href{mailto:du-zy23@mails.tsinghua.edu.cn}{du-zy23@mails.tsinghua.edu.cn}\\ \href{jense@zedat.fu-berlin.de}{jense@zedat.fu-berlin.de}\\ \href{cai.zhenyu.physics@gmail.com}{cai.zhenyu.physics@gmail.com}\\
\href{zwliu0@tsinghua.edu.cn}{zwliu0@tsinghua.edu.cn}}
\affiliation{Center for Quantum Information, Institute for Interdisciplinary Information Sciences, Tsinghua University, Beijing 100084, China}

\author{Jens Eisert}
\thanks{ZL and ZD contributed equally to this work.\\\href{qubithuan@gmail.com}{qubithuan@gmail.com}\\ \href{mailto:du-zy23@mails.tsinghua.edu.cn}{du-zy23@mails.tsinghua.edu.cn}\\ \href{jense@zedat.fu-berlin.de}{jense@zedat.fu-berlin.de}\\ \href{cai.zhenyu.physics@gmail.com}{cai.zhenyu.physics@gmail.com}\\
\href{zwliu0@tsinghua.edu.cn}{zwliu0@tsinghua.edu.cn}}

\affiliation{Dahlem Center for Complex Quantum Systems, Freie Universität Berlin, 14195 Berlin, Germany}

\affiliation{Helmholtz-Zentrum Berlin für Materialien und Energie, 14109 Berlin, Germany}

\author{Zhenyu Cai}
\thanks{ZL and ZD contributed equally to this work.\\\href{qubithuan@gmail.com}{qubithuan@gmail.com}\\ \href{mailto:du-zy23@mails.tsinghua.edu.cn}{du-zy23@mails.tsinghua.edu.cn}\\ \href{jense@zedat.fu-berlin.de}{jense@zedat.fu-berlin.de}\\ \href{cai.zhenyu.physics@gmail.com}{cai.zhenyu.physics@gmail.com}\\
\href{zwliu0@tsinghua.edu.cn}{zwliu0@tsinghua.edu.cn}}
\affiliation{Department of Engineering Science, University of Oxford, Parks Road, Oxford OX1 3PJ, United Kingdom}

\author{Zi-Wen Liu}
\thanks{ZL and ZD contributed equally to this work.\\\href{qubithuan@gmail.com}{qubithuan@gmail.com}\\ \href{mailto:du-zy23@mails.tsinghua.edu.cn}{du-zy23@mails.tsinghua.edu.cn}\\ \href{jense@zedat.fu-berlin.de}{jense@zedat.fu-berlin.de}\\ \href{cai.zhenyu.physics@gmail.com}{cai.zhenyu.physics@gmail.com}\\
\href{zwliu0@tsinghua.edu.cn}{zwliu0@tsinghua.edu.cn}}
\affiliation{Yau Mathematical Sciences Center, Tsinghua University, Beijing 100084, China}

\begin{abstract}

Many central tasks in fundamental physics and quantum information processing are possible only insofar as mixed quantum states can be made purer. 
In this work, we prove that the linearity and positivity of quantum mechanics impose general restrictions on quantum purification, unveiling a new fundamental principle of quantum information processing. 
We first establish that no quantum operation can transform a finite number of copies of an unknown quantum state or channel into an exactly pure output that depends non-trivially on the input, thereby ruling out an important form of universal purification in both static and dynamical settings. 
Building on this, we show that, upon relaxing the requirement of exact purity, one can establish quantitative sample-complexity lower bounds for approximate purification that hold for arbitrary physically allowed strategies, whose scaling matches the performance of purification-related tasks across several different areas of quantum information processing.
Moreover, this lower bound leads to a generalized standard quantum limit for learning arbitrary functions of a quantum state, greatly extending earlier results based on quantum Fisher information and revealing a deep connection between purification and quantum learning.
Extending this principle to other important settings, we establish, for the first time, an exponential sample-complexity lower bound for approximate pure dilation state preparation and a no-go theorem for approximate bosonic Gaussian state purification with passive Gaussian operations, establishing much more stringent limitations under practical operational constraints.

\end{abstract}
\maketitle


\section{Introduction}\label{sec:intro}

The fusion of quantum mechanics and information theory has unlocked unprecedented possibilities for information processing, enabling tasks that are impossible in the classical world and holding the potential for a revolutionary impact.
Yet, the fundamental principles of quantum mechanics impose strict boundaries on its capabilities.
Prominent examples include quantum cloning and deleting, tasks trivial in classical physics but forbidden by the linearity of quantum mechanics~\cite{Wootters1982,KumarPati2000deleting}, and quantum broadcasting, which is further limited by the positivity of quantum mechanics~\cite{barnum1996broadcasting,parzygnat2024vqb}.
Identifying such fundamental principles and the limitations they impose on operational tasks marks an essential route for shaping the understanding of the very nature of quantum mechanics as well as its technological potential.

The purification of noisy quantum systems, namely the reduction of quantum entropy in order to prepare a target pure system, is a fundamental primitive in quantum information science with both broad practical applications and deep theoretical significance. 
A standard and widely studied purification paradigm is to consume many copies of a noisy state or channel and concentrate their purity onto a single output copy~\cite{bennett1996distillation,cirac1999optimal,fiuraifmmode2004optimal}. 
On the practical side, such purification processes underlie a wide range of tasks, including quantum error correction and 
a number of variants of coherent quantum error
mitigation~\cite{terhal2015qec,cai2023qem,hugginsVirtualDistillationQuantum2021,koczorExponentialErrorSuppression2021}, quantum cooling~\cite{cotler2019cooling,zeng2021universal,andong2023linear}, and various resource distillation and concentration protocols~\cite{bennett1996distillation,Pan2001distillation,bravyikitaev,fang2020purification}. 
At the same time, purification also serves as a central conceptual framework in several foundational areas, including quantum Shannon theory~\cite{wilde2011classical}, resource theories~\cite{Horodecki2009entanglement,chitambar2019resource,liu19oneshot}, and quantum thermodynamics~\cite{Gour_2015}. 
Here, the number of input copies represents the cost of producing a pure output system, thus naturally characterizes the complexity of the task. 
A recurring theme in the literature is about analyzing the fundamental limitation of the number of input copies.

According to the relationship between the input and output of the purification process, purification tasks can be broadly divided into two categories.
The first category consists of tasks whose target output pure state or unitary is fixed and independent of the input. Prominent examples include entanglement distillation~\cite{bennett1996distillation} and magic state distillation~\cite{bravyikitaev}, where the usual target outputs are respectively a Bell state and a single-qubit $T$ state.
Limitations for such tasks originate from restrictions on the allowed physical maps, commonly referred to as free operations~\cite{liu19oneshot,fang2020purification,fang2022purification,Regula_2021,zang2025entanglement}, without which the task is trivialized by replacement maps that always output the desired target. 

\begin{figure}
\centering
\includegraphics[width=0.95\linewidth]{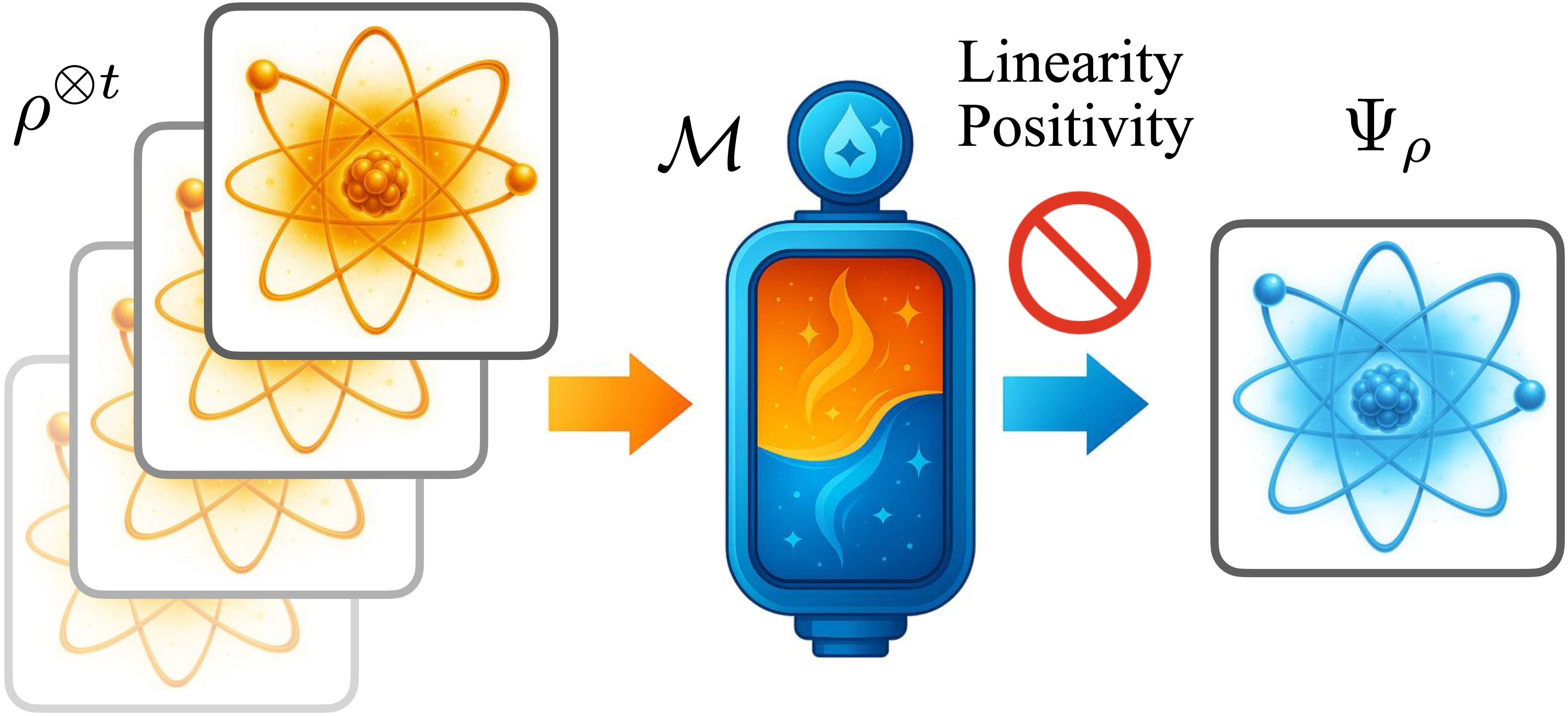}
\caption{The linearity and positivity of quantum mechanics fundamentally constrain the preparation of pure outputs that depend non-trivially on noisy inputs.
}
\label{fig:purification}
\end{figure}

The second and more fundamental category consists of input-dependent purification tasks, i.e., universal purification tasks, where the target output pure state or unitary itself depends on the input state or channel. Key examples include quantum purity amplification~\cite{li2024optimal,Childs2025streamingquantum} and quantum principal component analysis~\cite{Lloyd2014qpca}. Beyond these examples, input-dependent purification also fundamentally underlies notions of quantum learning and estimation. In particular, the preparation of a pure dilation state\footnote{In this work, we use ``pure dilation state'' to denote pure-state extension often referred to as ``purification'' in the literature, whereas ``purification'' represents a transformation process where the pure target state can be arbitrary.} of a mixed state $\rho$ using many copies of it, namely a global pure state whose reduced density matrix equals $\rho$, has important applications in learning nonlinear quantum properties~\cite{gilyen2019distributional,wang2023quantum,liu2024exponential,chen2024localtestunitarilyinvariant}. More generally, estimating physical properties of $\rho$ from many copies can itself be viewed as an input-dependent purification task, where the target pure output encodes the correct classical prediction of the target property.
Unlike the first category, the limitations of input-dependent purification tasks persist even without imposing any restrictions on the allowed operations. In other words, the obstruction is no longer simply a consequence of resource-theoretic constraints, and its physical and information-theoretic origin remains far from fully understood. Understanding the source of these limitations is therefore of fundamental importance, both for developing practical quantum purification protocols and for clarifying the ultimate limits of quantum learning.

In this work, we unveil a new physical principle: a general no-go principle on quantum purification, which rules out the possibility of purifying any finite copies of noisy states or channels into input-dependent pure states or unitaries.
This conclusion is fully protocol-independent: it applies to any physically admissible strategy, including sequential, coherent, and adaptive protocols.
With a short proof, we elucidate that this impossibility originates from two fundamental mathematical features of quantum mechanics---\emph{linearity} and \emph{positivity}---thus uncovering a new intrinsic limitation of quantum information processing, as demonstrated in Fig.~\ref{fig:purification}.
Beyond its theoretical significance, this limitation directly constrains the performance of practical quantum tasks. 
By relaxing the requirement of producing an exactly pure state to allowing nearly pure outputs, we derive a general sample complexity lower bound of approximate purification. 
The lower bound meets the performance of various tasks, including quantum purity amplification and quantum principal component analysis algorithms.
Importantly, the performance of approximate purification implies a standard quantum limit for estimating arbitrary functions of unknown quantum systems, greatly enhancing the former result derived using quantum Fisher information and showing the deep connection between quantum purification and learning.
Beyond these general conclusions, we show that these fundamental limitations become significantly more stringent when specific operational constraints are imposed in several important tasks.
For the preparation of a pure dilation state, we prove an exponential sample complexity lower bound.
For the purification of bosonic Gaussian states, we show the no-go result with passive Gaussian operations, even for approximate purification.

\section{No universal exact purification}

A universal quantum purification task involves finding a quantum process, represented by $\mathcal{M}$, that can map many copies of an arbitrary noisy input system into a pure state or unitary depending on the input.
At the same time, the quantum process, $\mathcal{M}$, is universal and independent of the input system.
The starting point of our framework can be summarized in the following observation, which very crisply showcases the fundamental obstructions to purification.

\begin{observation}[No universal exact purification]\label{obs:no-purification}
Let $t$ be any finite integer, and let 
$\mathcal{M}$ be a nonzero positive map which maps positive semi-definite matrices to positive semi-definite matrices. 
Suppose 
that for any state $\rho$ supported on a subspace of the space of all density matrices $\mathcal{D}_S \subseteq \mathcal{D}(\mathcal{H}_d)$ with $\dim(\mathcal{D}_S)\ge 2$, the output $\mathcal{M}(\rho^{\otimes t})$ is proportional to a pure state. Then, for any $\rho\in\mathcal{D}_S$, the output of $\mathcal{M}$ must always be proportional to a fixed pure state,
\begin{equation}
\mathcal{M}(\rho^{\otimes t})\propto\Psi.
\end{equation}
\end{observation}
\begin{proof}
According to the assumption made by this observation, when $\dim(\mathcal{D}_S)\ge2$, there exist two fixed and linearly independent density matrices $\rho_0$ and $\rho_1$, such that for any $0\le x\le1$, we have
\begin{equation}\label{eq:ab}
\mathcal{M}\left[(x\rho_0+(1-x)\rho_1)^{\otimes t}\right]\propto\Psi_x.
\end{equation}
Because the input state can be written as 
\begin{equation}
\sum_{k=0}^t \tbinom{t}{k}x^k(1-x)^{t-k}\Pi_{\mathrm{sym}}(\rho_0^{\otimes k}\otimes\rho_1^{\otimes(t-k)}),
\end{equation}
with $\Pi_{\mathrm{sym}}$ being the symmetrization channel ($\Pi_{\mathrm{sym}}(\rho_0\otimes\rho_1)=\frac{1}{2}\rho_0\otimes\rho_1+\frac{1}{2}\rho_1\otimes\rho_0$), the output state can thus be represented as
\begin{equation}\label{eq:expansion}
\sum_{k=0}^t \tbinom{t}{k}x^k(1-x)^{t-k}\mathcal{M}\left[\Pi_{\mathrm{sym}}(\rho_0^{\otimes k}\otimes\rho_1^{\otimes(t-k)})\right],
\end{equation}
due to \emph{linearity} of $\mathcal{M}$.
According to Eq.~\eqref{eq:ab}, we have
\begin{equation}
\sum_{k=0}^t \tbinom{t}{k}x^k(1-x)^{t-k}X_k\propto\Psi_x,
\end{equation}
where $X_k=\mathcal{M}\left[\Pi_{\mathrm{sym}}(\rho_0^{\otimes k}\otimes\rho_1^{\otimes(t-k)})\right]$ is a positive semi-definite matrix as $\mathcal{M}$ is a \emph{positive} map.
Notice that the l.h.s.\ of the above equation is a linear combination of positive semi-definite matrices with positive coefficients, while the r.h.s.\ is a rank-$1$ matrix. As the r.h.s.\ is rank-$1$, every positive semi-definite matrix appearing in the l.h.s.\ must have support contained in the same one-dimensional subspace. This means that each positive matrix should be proportional to the rank-$1$ matrix on the r.h.s., including the first and the last matrices corresponding to $k=0$ and $k=t$.
We thus have 
\begin{equation}
\mathcal{M}(\rho_0^{\otimes t})\propto\Psi_x, \ \ \ \mathcal{M}(\rho_1^{\otimes t})\propto\Psi_x,
\end{equation}
for arbitrary $\rho_0$ and $\rho_1$, which concludes our proof.
\end{proof}

Thus, any universal exact purification protocol must be trivial: whenever the output is required to be pure for all inputs, it cannot depend non-trivially on the input state.
Although the assumption of producing an exactly pure output state is idealized and not directly achievable in realistic settings, this observation identifies the fundamental physical origin of the no-purification phenomenon. Specifically, the proof relies solely on the linearity and positivity of $\mathcal{M}$, two fundamental structural principles of quantum mechanics, which also underlie all later no-go results on approximate purification derived in this work. Moreover, since we do not require $\mathcal{M}$ to be trace-preserving, the result applies to all quantum operations, including post-selection and adaptive strategies.

Extending beyond quantum states, our no-purification result can be naturally generalized to quantum channels and higher-order maps~\cite{liu2020channel}. 
In particular, there is no universal procedure that transforms $t$ copies of an unknown input object into an output object with a pure Choi representation that depends non-trivially on the input. 
This excludes, for example, universal transformations from quantum channels to unitary channels, from quantum states to unitary channels, and from quantum channels to pure states. 
Moreover, the conclusion is independent of the specific protocol structure: it applies to arbitrary physically allowed strategies, including sequential, coherent, and adaptive protocols. 
The key reason is that any physically admissible supermap acting on channels can be represented, at the level of Choi operators, as a completely positive trace-preserving map acting on the corresponding Choi states~\cite{Chiribella2008superchannel,chiribella2008comb,liu2024non}. 
Since the Choi state of a unitary channel is pure, any such hypothetical transformation would induce a quantum channel mapping input Choi states to input-dependent pure states, directly contradicting Observation~\ref{obs:no-purification}.

\section{Approximate purification}
In the above, we have addressed the fundamental restriction on perfect quantum purification, whose targets are pure states or unitaries.
Nonetheless, approximate pure systems are often already useful in many tasks, thus the characterization of an approximate version of the principle is crucial for practical purposes.
With full details in Appendix~\ref{app:proof_approximate_purification}, we analyzed the robustness of Observation~\ref{obs:no-purification}, showing a clear trade-off relationship.
Specifically, for any positive trace-preserving map $\mathcal{M}$ with a fixed $t$, if its output state $\mathcal{M}(\rho^{\otimes t})$ is close to a pure state, then the difference between its output states will be small.
When the purification task demands the output state to vary significantly according to the input state, then one requires a large number of input states to obtain high purity.
This interesting trade-off relationship 
is encapsulated in the following theorem.

\begin{theorem}[Approximate purification]\label{thm:approx_no-purification_distance_between_states-new}
Let $\mathcal{M}$
be a positive trace-preserving map. Suppose that there exists a parameter $0 \le \epsilon < 1$ such that for all $\rho$ supported on a subspace $\mathcal{D}_S\subseteq\mathcal{D}(\mathcal{H}_d)$ with $\mathrm{dim}(\mathcal{D}_S)\ge 2$ 
\begin{equation}
    \Tr\bigl[\mathcal{M}(\rho^{\otimes t})^2\bigr] \ge 1-\epsilon.
\end{equation}
Then there exists a state $\sigma_0$ such that, for all $\rho\in\mathcal{D}_S$,
\begin{equation}
    D\bigl(\sigma_0,\mathcal{M}(\rho^{\otimes t})\bigr) \le (8 \pi\sqrt{2t} + 12) \sqrt{\epsilon}.
\end{equation}
Here $D(\cdot,\cdot)$ denotes the trace distance $D(\omega,\tau)=\tfrac{1}{2}\Vert \omega-\tau\Vert_1$.
\end{theorem}

Here, we let $\mathcal{M}$ be trace-preserving since we are considering the sample complexity; that is, $\mathcal{M}$ should output the target state with unit probability. 
The proof, left to Appendix~\ref{app:proof_approximate_purification}, is a generalization of the proof of Observation~\ref{obs:no-purification}.
By the same Choi-state reduction, the approximate purification bound in Theorem~\ref{thm:approx_no-purification_distance_between_states-new} also extends to all channel-involving settings and applies to arbitrary physically allowed strategies.
Specifically, if one aims to find a universal quantum protocol that approximately transforms $t$ copies of quantum channel into a unitary or pure quantum state, or transforms $t$ copies of quantum states into a unitary, similar trade-off exists.

The fundamental limitation of approximate universal purification has several important consequences. 
In particular, whenever the outputs corresponding to two different input states remain distinguishable, namely $D\left(\mathcal{M}(\rho_1^{\otimes t}),\mathcal{M}(\rho_2^{\otimes t})\right)=\Omega(1)$, Theorem~\ref{thm:approx_no-purification_distance_between_states-new} immediately implies the sample complexity lower bound
\begin{equation}
t=\Omega(\epsilon^{-1}).
\end{equation}
This scaling reveals a fundamental distinction between universal purification and purification tasks with a fixed target output state. 
For the latter class of tasks with restrictions from the free operations, the purification complexity bound scales only as $\log(\epsilon^{-1})$~\cite{fang2020purification,fang2022purification}.

Interestingly, the $\epsilon^{-1}$ scaling appears universally across a broad range of quantum information processing tasks. 
Examples include quantum purity amplification~\cite{li2024optimal,Childs2025streamingquantum}, which aims to prepare the eigenstate of $\rho$ corresponding to its largest eigenvalue, and quantum principal component analysis~\cite{Lloyd2014qpca}, which aims to transform many copies of $\rho$ into the unitary evolution $e^{-i\rho}$. 
Detailed analyses showing why these tasks achieve the complexity scaling $\mathcal{O}(\epsilon^{-1})$ are provided in Appendix~\ref{app:proof_purity_amplification} and Appendix~\ref{app:proof_qpca}. 
Notably, the same $\epsilon^{-1}$ behavior also appears in covariant codes~\cite{faist2020aqec,kong2022aqec} and quantum channel purification~\cite{Ryotaro2026channelpurification}, which aim to prepare the ideal unitary channel using several noisy ones. 
Together, these important examples suggest that approximate universal purification provides a common information-theoretic framework underlying several previously unrelated complexity phenomena in quantum information processing.

Another fundamental consequence of approximate purification arises in quantum learning and estimation. 
Indeed, estimating a property of an unknown quantum state $\rho$, represented by a function $f(\rho)$, from many copies of $\rho$ can itself be viewed as a universal purification task. 
This connection can be established through a simple measure-and-prepare protocol. Given $t$ copies of $\rho$, one first constructs an estimator $\hat{f}$ for $f(\rho)$ using arbitrary quantum measurements, and then prepares the single-qubit  state vector
\begin{equation}
| {\psi_{\hat{f}}} \rangle=\frac{1}{\sqrt{2}}\left(\ket{0}+e^{i\hat{f}}\ket{1}\right).
\end{equation}
This procedure naturally fits into the framework of universal purification: the protocol itself is independent of the input state, while the target output pure state depends non-trivially on $\rho$ through the quantity $f(\rho)$. 
However, for any finite number of copies $t$, the estimator $\hat{f}$ inevitably exhibits statistical fluctuations due to the intrinsic probabilistic nature of quantum mechanics. Consequently, the actual output of the measure-and-prepare channel is the mixed state
\begin{equation}
\mathbb{E}\!\left[|{\psi_{\hat{f}}}\rangle\langle{\psi_{\hat{f}}}|\right],
\end{equation}
whose purity is directly related to the variance of $\hat{f}$. 
Using this observation, we derive a generalized standard quantum limit for estimating arbitrary functions of $\rho$, with the detailed proof presented in Appendix~\ref{app:proof_sql}.

\begin{corollary}[Standard quantum limit for arbitrary function]\label{coro:sql}
Given $t$ copies of unknown quantum state $\rho$, which is supported on a subspace $\mathcal{D}_S\subseteq\mathcal{D}(\mathcal{H}_d)$ with $\mathrm{dim}(\mathcal{D}_S)\ge 2$, any target function $f(\rho)$, and an estimator $\hat{f}$ constructed by any measurement performed on $t$ copies of $\rho$ with 
\begin{equation}\label{eq:sql_difference}
\sup_{\rho}\abs{f(\rho)-\mathbb{E}[\hat{f}]}\le \frac{1}{10}\left[\sup_\rho f(\rho)-\inf_\rho f(\rho)\right]=\frac{1}{10}\Delta_f,
\end{equation}
then the variance of $\hat{f}$ is lower bounded by
\begin{equation}\label{eq:sql_variance}
   \sup_\rho \mathrm{Var}\left(\hat{f}\right)=\Omega(\Delta_f^2t^{-1}).
\end{equation}
\end{corollary}
This result also directly generalizes to the estimation of channel properties and substantially generalizes previous derivations in quantum metrology based on the additivity of quantum Fisher information~\cite{braunstein1994geometry,Tsang2020semi}, as it does not require the target function to satisfy any smoothness or differentiability assumptions and applies to global estimation. 
It shows that the standard quantum limit is an extremely fundamental principle for quantum learning, while also demonstrating the tightness of the lower bound established in Theorem~\ref{thm:approx_no-purification_distance_between_states-new}.

\section{Case Studies}\label{sec:case}

\subsection{Approximate pure dilation state preparation}

The $\Omega(\epsilon^{-1})$ sample complexity lower bound is the most general lower bound that holds for all purification tasks, which is dimension-independent and does not depend on the details of the purification task.
The only requirements are: (i) the output should be close to a pure state, and (ii) with different inputs, the outputs should have a large deviation.
By contrast, other practical tasks with more stringent requirements may require a much higher sample complexity, such as the preparation of a pure dilation state.
We define 
$\Psi_{A,B}$ as the pure dilation state of a mixed state $\rho$ 
so that
$\Tr_B(\Psi_{A,B})=\rho$.
Besides being a fundamental concept, the pure dilation state also carries broad practical significance in areas such as quantum algorithms, learning, and simulation.
In Ref.~\cite{liu2024exponential}, the authors show that with pure dilation states, certain quantum non-linear properties can be efficiently estimated without a high requirement for quantum memory.
A prevalent quantum algorithmic framework, known as \emph{quantum singular value transformation} (QSVT) ~\cite{Low2019hamiltonian,martyn2021grand}, is based on an algorithm that encodes the target matrix as a block of a unitary operation.
If one encodes the target density matrix into a unitary, one can efficiently estimate properties such as mixed-state fidelity or quantum entropy~\cite{gilyen2022petz,wang2023fidelity,wang2024entropy}.
These algorithms typically rely on an oracle to prepare the pure dilation state of the target quantum state~\cite{gilyen2019distributional}.
Hence, analyzing the complexity in preparing a pure dilation state helps to understand the practical difficulty in realizing such an oracle.

By definition, this task aims to find a quantum operation, $\mathcal{M}$, that acts on $\rho^{\otimes t}$ and the output state is some pure dilation state satisfying two requirements: (1) the output is a pure state; (2) the marginal of the output state is $\rho$.
Theorem~\ref{obs:no-purification} shows that it is impossible to simultaneously achieve these two requirements, but in practice, it is also valuable to understand if we can do so up to some approximation. 
\begin{theorem}\label{thm:no_approx_dilation}
Let $\mathcal{M}:\mathcal{D}(\mathcal{H}_A^{\otimes t})\to\mathcal{D}(\mathcal{H}_{A,B})$ be a positive trace-preserving map with $d \ge C$ for a sufficiently large
absolute constant $C$.
Suppose there exist parameters $\epsilon,\eta\ge 0$ such that $\xi \coloneqq 2(\sqrt{\epsilon}+\epsilon+\eta) < \tfrac{1}{3}$ and, for arbitrary $\rho$,
\begin{enumerate}
    \item[(1)] (Near purity) $\Tr\left[\mathcal{M}(\rho^{\otimes t})^2\right]\ge1 - \epsilon$,
    \item[(2)] (Marginal accuracy) $D\left(\Tr_B[\mathcal{M}(\rho^{\otimes t})],\rho\right)\le\eta$.
\end{enumerate}
Then
\begin{equation}\label{eq:dilation_lower_bound}
    t \ge \frac{ed}{128 \left[\xi + 13\exp\left(-\frac{d}{32(\ln d)^2}\right)\right]^2}.
\end{equation}
\end{theorem}
For qubit systems, the linear scaling with $d$ in Eq.~\eqref{eq:dilation_lower_bound} implies that the sample complexity $t$ scales exponentially with the system size $n$ (where $d = 2^n$). 
This exponential overhead rules out the efficient preparation of a pure dilation state for the aforementioned learning or algorithmic tasks.
The proof follows from a direct extension of the proof of Observation~\ref{obs:no-purification}, with the full proof provided in Appendix~\ref{app:no_approx_purification}.
Following the same argument, if the output state of $\mathcal{M}(\rho^{\otimes t})$ is close to the pure dilation state of $\rho$, then, given any positive decomposition, $\rho^{\otimes t}=\sum_ip_iM_i$ with $p_i\ge 0$ and $M_i$ being positive semi-definite matrix, $\mathcal{M}(M_i)$ should also be close to the pure dilation state of $\rho$.
Suppose the target state is given by some classical probability distribution $\textbf{p}=\{p_i\}_i$, $\rho=\sum_i p_i\ketbra{i}{i}$, $\rho^{\otimes t}$ can be decomposed into different orthogonal states with probability $\textbf{p}^t$.
Thus, if the map for preparing an approximate pure dilation state exists, it can map a single state sampled by $\textbf{p}^t$ to the pure dilation state of $\rho$ and recover the whole probability distribution $\textbf{p}$ by tracing a subsystem.
According to a well-known result in distribution estimation~\cite{Han2015Minimax}, we can establish the exponential lower bound for $t$.

\subsection{Gaussian purification}

In all previous discussions, we imposed no restrictions on the universal purification map itself: any physically allowed quantum operation was permitted. As a consequence, approximate universal purification is always achievable in principle. Indeed, simple measure-and-prepare protocols already provide such constructions --- once some information about the input state can be learned from many copies, one can always encode this information into an approximate pure output state. This naturally raises a more fundamental question: can both exact and approximate universal purification become impossible once physical restrictions on the allowed quantum operations are imposed?

In this section, we investigate this question using passive Gaussian quantum channels~\cite{cerf2007quantum,GaussianQuantumInfo,caruso2011optimaldilation}. Because these operations arise naturally in linear optical implementations~\cite{cerf2007quantum}, they provide a standard framework for studying bosonic quantum systems and quantum optical platforms. We will show that, when restricting to only these operations, even approximate universal purification encounters fundamental limitations. Since this is no longer a finite-dimensional setting, the relevant preliminaries and proofs are deferred to Appendix~\ref{app:Gaussian}. We now begin by presenting the analogue of no universal exact purification for passive Gaussian channels.

\begin{observation}[No Gaussian universal exact purification]\label{obs:no-Gaussian-purification}
Let $t$ be any finite integer, and let 
$\mathcal{M}: \mathcal{D}(\mathcal{H}^{\otimes t})
\rightarrow \mathcal{D}(\mathcal{H})$ be a passive Gaussian channel. Suppose 
that for any Gaussian state $\rho$ supported on $\mathcal{H}$, the output $\mathcal{M}(\rho^{\otimes t})$ is a pure 
Gaussian state. Then, for any Gaussian input state $\rho$, the output of $\mathcal{M}$ must always be a fixed pure Gaussian state,
$\mathcal{M}(\rho^{\otimes t})=\Psi$.
\end{observation}

Observation~2 shows that passive Gaussian operations cannot realize exact universal purification. The next question is whether this limitation can be overcome approximately by increasing the number of input copies. Surprisingly, unlike the unrestricted setting where approximate purification can always be achieved through measure-and-prepare strategies, passive Gaussian operations remain fundamentally incapable of universal purification even in the approximate regime.
\begin{theorem}[No Gaussian universal approximate purification]\label{thm:approx_no-purification_distance_between_states-Gaussian}
Let $\mathcal{M}: \mathcal{D}(\mathcal{H}^{\otimes t})
\rightarrow \mathcal{D}(\mathcal{H})$
be a passive Gaussian channel. Suppose that there exists a parameter $0 \le \epsilon < 1$ such that for all Gaussian states 
$\rho$ contained in a subspace ${\cal G}$ of Gaussian states with covariance matrices that are larger than 
$\mathbb{I}$ in matrix ordering, 
we have that
\begin{equation}\label{BigP}
    {\rm Tr} 
    \bigl[\mathcal{M}(\rho^{\otimes t})^2\bigr] \ge 1-\epsilon.
\end{equation}
Then there exists a Gaussian state $\Psi$ such that, for all Gaussian states 
$\rho\in {\cal G}$,
 \begin{equation}\label{eq:gaussian_distance}
 	D(\mathcal{M}(\rho^{\otimes t}), \Psi) \leq 4(1+\sqrt{3}) \epsilon (1+4 \epsilon).
 \end{equation}
\end{theorem}
Comparing this result with Theorem~\ref{thm:approx_no-purification_distance_between_states-new} or Theorem~\ref{thm:no_approx_dilation}, we can notice the main difference: the r.h.s.\ of Eq.~\eqref{eq:gaussian_distance} is 
independent of $t$, indicating that coherent state manipulation does not bring an advantage. This is to an extent reminiscent of earlier findings in a different context that show that coherent state manipulation of Gaussian states finds obstructions \cite{PhysRevLett.89.137903,PhysRevA.66.032316,PhysRevLett.89.137904,GiedkeSchmidt}. 
This leads to an important consequence: passive Gaussian operation cannot perform an 
arbitrary input-dependent purification task, even approximately. Note also that in the purification task, the subset of inputs we consider is that of Gaussian states the variance of which is not below that of the vacuum for any quadrature, which is natural in a purification setting.

\section{Discussion and outlook}

Identifying the fundamental no-go principles of quantum information processing has been one of the central and defining goals of quantum information science.
A common paradigm for establishing such results is to first prove an exact impossibility theorem in an idealized setting, then study the achievable performance and resource requirements in the approximate regime, and finally explore the implications in specific scenarios of practical relevance.
In this work, we have completed this program for the fundamental problem of universal purification. Observation~\ref{obs:no-purification} establishes the no-go theorem in the exact setting, while Theorem~\ref{thm:approx_no-purification_distance_between_states-new} characterizes its asymptotic behavior in the approximate regime.
Furthermore, Corollary~\ref{coro:sql} and Theorem~\ref{thm:no_approx_dilation} reveal deep connections between approximate universal purification and quantum learning.
Finally, Observation~\ref{obs:no-Gaussian-purification} and Theorem~\ref{thm:approx_no-purification_distance_between_states-Gaussian} show that, under physically motivated restrictions on the allowed operations, namely passive Gaussian operations, universal purification can remain fundamentally impossible even in the approximate regime.

These results constrain natural tasks in coherent quantum error mitigation~\cite{terhal2015qec,cai2023qem,hugginsVirtualDistillationQuantum2021,koczorExponentialErrorSuppression2021} and quantum cooling~\cite{cotler2019cooling,zeng2021universal,andong2023linear}.
In particular, our results represent an information-theoretic counterpart of the third law of thermodynamics, showing that beyond the well known impossibility of cooling to zero temperature with finite thermodynamic processes, no positive map whatsoever can universally cool quantum states to their corresponding zero-temperature pure states of arbitrary systems with finite resources. Moreover, our approximate purification results directly imply fundamental complexity bounds for low-temperature state preparation, which can be viewed as quantitative manifestations of the above principle and complements existing thermodynamic bounds~\cite{ThirdLaw}.
In this sense, we hope that our work will stimulate new thinking in quantum information processing: rather than proposing new schemes for purification, coherent error mitigation, or cooling, we provide tools for fundamentally assessing how well such tasks can perform.

Interestingly, it is possible to sidestep certain no-purification limitations in specific settings. 
For example, in some quantum learning tasks such as quantum tomography, one does not necessarily require a pure dilation state itself as the final resource. 
Recent works on \emph{random purifications}~\cite{pelecanos2025mixed,tang2025conjugate,mele2025random,girardi2025random,yoshida2025random,mele2025optimal} exploit this observation by preparing an ensemble of random pure dilation states, which effectively corresponds to a mixed-state resource rather than a single pure dilation state. 
In this way, the restriction of Theorem~\ref{thm:no_approx_dilation} can be bypassed for certain quantum learning applications.
More broadly, an important open question is whether similar mixed-state substitutes can be found in other scenarios requiring universal purification, or whether the limitations imposed by no purification can be sidestepped by paying additional computational costs or by relaxing the requirements of the target task. 
A related example already appears in quantum error mitigation, where non-positive maps are employed to improve the accuracy of expectation-value estimation at the cost of increased sampling complexity~\cite{cai2023qem}.

Another likely fruitful direction is to further investigate the relationship between universal purification and quantum learning.
The proof of Theorem~\ref{thm:no_approx_dilation} is achieved by reducing the preparation of a pure dilation state to the learning of a classical probability distribution. A natural next question is whether this reduction can be strengthened to quantum distributions, potentially showing that the complexity of preparing pure dilation states is fundamentally equivalent to that of quantum tomography itself. 
Moreover, Corollary~\ref{coro:sql} demonstrates through a measure-and-prepare construction that limitations on approximate purification imply limitations for quantum learning. 
From this perspective, the restriction on passive Gaussian operations established in Theorem~\ref{thm:approx_no-purification_distance_between_states-Gaussian} may indicate the existence of fundamental limitations on Gaussian-state learning using passive Gaussian operations. 
Understanding these possible connections between purification, tomography, and learning 
complexity may provide a deeper picture of the ultimate capabilities and limitations of quantum information processing. It is the hope that the present work inspires some of these interesting research
directions and inspires new thinking about purification tasks.

\section*{Acknowledgments}
Z.L.\ appreciates Ryuji Takagi and Kento Tsubouchi for inspiring him to consider the problem of preparing pure dilation states.
We also appreciate insightful discussions with  Qisheng Wang, Yuxiang Yang, Yunlong Xiao, Masahito Hayashi, and Allen Zang.
Z.L.\ and Z.-W.L.\ are supported in part by NSFC under Grant No.~12475023, Dushi Program, and a startup funding from YMSC.
Z.D.\ acknowledges the support from the National Natural Science Foundation of China Grant No.~12174216 and the Innovation Program for Quantum Science and Technology Grant No.~2021ZD0300804 and No.~2021ZD0300702.
Z.C.\ acknowledges support from the EPSRC project Robust and Reliable Quantum Computing (RoaRQ, EP/W032635/1) and the EPSRC quantum technologies career acceleration fellowship (UKRI1226). Work in Berlin has been funded by the 
BMFTR (Hybrid++, QoSol, PasQuops, MuniQCAtoms), the Quantum Flagship (Millenion, PasQuans2), Berlin Quantum, the Munich Quantum Valley, the DFG (SPP 2514), and the European Research Council (DebuQC).

\appendix

\section{Approximate purification}\label{app:approx_purification}

\subsection{Proof of Theorem~\ref{thm:approx_no-purification_distance_between_states-new}\label{app:proof_approximate_purification}}

Fix an integer $t$ and two states $\rho_0, \rho_1 \in \mathcal{D}(\mathcal{H}_S)$. For $x \in [0, 1]$, define the mixture $\rho_x \coloneqq (1-x)\rho_0 + x\rho_1$. Write $\sigma_x\coloneqq\mathcal{M}(\rho_x^{\otimes t})$ as a mixture of states as
\begin{equation}
\sigma_x = \sum_{k=0}^{t} \binom{t}{k} x^k (1-x)^{t-k} \omega_k = \sum_{k=0}^{t} p_x(k) \omega_k,
\end{equation}
where $p_x(k)\coloneqq \binom{t}{k}x^{k}(1-x)^{t-k}$ and $\omega_k \coloneqq \mathcal{M}\bigl(\Pi_{\mathrm{sym}}(\rho_0^{\otimes t-k} \otimes \rho_1^{\otimes k})\bigr)$.

Define $\phi_x$ to be the corresponding eigenvector of the largest eigenvalue $\lambda_1$ of $\sigma_x$. Then, 
\begin{equation}\label{eq:closeness-eigenvalue-mixed-state}
    \bra{\phi_x}\sigma_x\ket{\phi_x} = \lambda_1 = \lambda_1(\sum_i \lambda_i) \ge \sum_i \lambda_i^2 \ge 1-\epsilon,
\end{equation}
where the last inequality follows from the near-purity condition. 
To analyze the components $\omega_k$, we partition the indices into ``good'' and ``bad'' sets based on their overlap with the dominant eigenvector $\phi_x$ as
\begin{equation}\label{eq:fidelity-eigenvalue-component}
\mathcal{G}_x \coloneqq \{k : \langle \phi_x | \omega_k | \phi_x \rangle \ge 1 - 4\epsilon \}, \quad \mathcal{B}_x \coloneqq \{0, \dots, t\} \setminus \mathcal{G}_x.
\end{equation}
We first prove the following lemmas.

\begin{lemma}[Closeness of eigenvalues]\label{lem:closeness-eigenvalue}
For any $x, y \in [0,1]$, if $\mathcal{G}_x \cap \mathcal{G}_y \neq \emptyset$, then $D(\phi_x,\phi_y) \le 4 \sqrt{\epsilon}$.
\end{lemma}

\begin{proof}
Pick an index $k \in \mathcal{G}_x \cap \mathcal{G}_y$. By the definition of $\mathcal{G}_x$ in Eq.~\eqref{eq:fidelity-eigenvalue-component}, we have $\langle \phi_x | \omega_k | \phi_x \rangle \ge 1-4\epsilon$. Applying the Fuchs–van de Graaf inequality for a pure state and a general density matrix, we obtain
\begin{equation}
D(\phi_x, \omega_k) \le \sqrt{1 - \langle \phi_x | \omega_k | \phi_x \rangle} \le \sqrt{4\epsilon} = 2\sqrt{\epsilon}.
\end{equation}
Similarly, $D(\phi_y, \omega_k) \le 2\sqrt{\epsilon}$ also holds. The claim then follows directly from the triangle inequality.
\end{proof}

\begin{lemma}[Conditional intersection]\label{lem:condition-intersection}
For any $x, y \in [0, 1]$, if $|y-x| < \sqrt{\frac{x(1-x)}{2t}}$, then $\mathcal{G}_x \cap \mathcal{G}_y \neq \emptyset$.
\end{lemma}

\begin{proof}
    From Eq.~\eqref{eq:closeness-eigenvalue-mixed-state}, we have $\sum_k p_x(k) \langle \phi_x | \omega_k | \phi_x \rangle \ge 1-\epsilon$. By Markov's inequality, the probability of the bad set 
    satisfies
\begin{equation}\label{eq:prob_bad_set}
    \sum_{k \in \mathcal{B}_x} p_x(k) \le \frac{\sum_k p_x(k) (1 - \langle \phi_x | \omega_k | \phi_x \rangle)}{4\epsilon} \le \frac{1}{4}.
    \end{equation}
    Similarly, $\sum_{k \in \mathcal{B}_y} p_y(k) \le 1/4$.

    Let $q(k) \coloneqq \min(p_x(k), p_y(k))$. By definition, $\sum_{k} q(k) = 1 - \TV(p_x, p_y)$. The sum over the union of bad sets is bounded by
    \begin{equation}
        \sum_{k \in \mathcal{B}_x \cup \mathcal{B}_y} q(k) \le \sum_{k \in \mathcal{B}_x} p_x(k) + \sum_{k \in \mathcal{B}_y} p_y(k) \le \frac{1}{2}.
    \end{equation}
    Thus, the sum over the intersection of good sets satisfies
    \begin{equation}
    \begin{split}
        \sum_{k \in \mathcal{G}_x \cap \mathcal{G}_y} q(k) &= \sum_{k} q(k) - \sum_{k \in \mathcal{B}_x \cup \mathcal{B}_y} q(k) \\
        &\ge \left(1 - \TV(p_x, p_y)\right) - \frac{1}{2} \\
        &= \frac{1}{2} - \TV(p_x, p_y).
    \end{split}
    \end{equation}
    To bound $\TV(p_x, p_y)$, 
    we use Pinsker's inequality and the properties of the binomial distribution
    \begin{equation}
    \begin{split}
        \TV(p_x, p_y) &\le \sqrt{\frac{1}{2} \mathrm{KL}( \mathrm{Bin}(t,y) \| \mathrm{Bin}(t,x) )} \\
        &= \sqrt{\frac{t}{2} \mathrm{KL}( \mathrm{Ber}(y) \| \mathrm{Ber}(x) )}.
    \end{split}
    \end{equation}
Using the upper bound  \begin{equation}
\mathrm{KL}(\mathrm{Ber}(y) \| \mathrm{Ber}(x)) \le \frac{(y-x)^2}{x(1-x)} 
 \end{equation}
 for the Bernoulli distribution, and substituting the assumption $|y-x| < \sqrt{\frac{x(1-x)}{2t}}$, we obtain
\begin{equation}
    \TV(p_x, p_y) < \sqrt{\frac{t}{2} \cdot \frac{1}{2t}} = \frac{1}{2}.
\end{equation}
It follows that $\sum_{k \in \mathcal{G}_x \cap \mathcal{G}_y} q(k) > \frac{1}{2} - \frac{1}{2}= 0$, which implies $\mathcal{G}_x \cap \mathcal{G}_y \neq \emptyset$.
\end{proof}

\begin{lemma}[Increasing distance]\label{lem:increasing-distance}
Let $x_0 = \frac{1}{t+1}$ and define the sequence $x_{i+1} = \min(x_i + \Delta_i, 1)$ with step size $\Delta_i \coloneqq \kappa \sqrt{\frac{x_i(1-x_i)}{t}}$. For $m \ge \frac{\pi \sqrt{1+2\kappa}}{\kappa} \sqrt{t} + 1$, we have $x_m = 1$.
\end{lemma}

\begin{proof}
Assume for contradiction that $x_i + \Delta_i < 1$ for all $i < m$, such that $x_{i+1} - x_i = \Delta_i$. Define the potential function $u(x) \coloneqq \arcsin(2x - 1)$ for $x \in [0, 1]$ and let $h(x) \coloneqq x(1-x)$. Note that $u'(x) = 1/\sqrt{h(x)}$.
For $x_i < x_{i+1}\le 1/2$, by the mean value theorem, there exists $\xi_i \in [x_i, x_{i+1}]$ such that
\begin{equation}
    u(x_{i+1}) - u(x_i) = u'(\xi_i) \Delta_i = \frac{\Delta_i}{\sqrt{h(\xi_i)}}.
\end{equation}
Since $h(x)$ as a function of $x$ is increasing on $(0, 1/2]$, we have $u(x_{i+1}) - u(x_i) \ge \Delta_i / \sqrt{h(x_{i+1})}$. To bound $h(x_{i+1})$, observe that $h(x_{i+1}) = h(x_i) + (1-2x_i)\Delta_i - \Delta_i^2 \le h(x_i) + \Delta_i$. Given $x_i \ge x_0$, we have 
\begin{equation}
\sqrt{t h(x_i)} \ge \sqrt{t \frac{t}{(t+1)^2}} \ge 1/2 
\end{equation}
for $t \ge 1$. Thus, $\Delta_i = \kappa h(x_i) \frac{1}{\sqrt{t h(x_i)}} \le 2\kappa h(x_i)$. It follows that $h(x_{i+1}) \le h(x_i)(1+2\kappa)$. Substituting this into the lower bound for the potential change, we get
\begin{equation}
\begin{split}
    u(x_{i+1}) - u(x_i) &\ge \frac{\kappa \sqrt{h(x_i)/t}}{\sqrt{h(x_i)(1+2\kappa)}} \\
    &= \frac{\kappa}{\sqrt{t(1+2\kappa)}}.
\end{split}
\end{equation}
For $x_i \ge 1/2$, a similar argument holds using the decreasing monotonicity of $h(x)$, yielding $u(x_{i+1}) - u(x_i) \ge \kappa/\sqrt{t}$.

Accounting for that at most one step where the interval $[x_i, x_{i+1}]$ contains $1/2$, the total change after $m$ steps satisfies
\begin{equation}
\begin{split}
    u(x_m) - u(x_0) &= \sum_{i=0}^{m-1} (u(x_{i+1}) - u(x_i)) \\
    &\ge (m-1) \frac{\kappa}{\sqrt{t(1+2\kappa)}},
\end{split}
\end{equation}
the maximum range of the 
potential function is $u(1) - u(0) = \pi/2 - (-\pi/2) = \pi$. Since $u(x_0) > u(0)$, this implies $u(x_m) > u(1)$, a contradiction. Hence, the sequence must reach $x_m = 1$.
\end{proof}

Now we proceed to proving
Theorem~\ref{thm:approx_no-purification_distance_between_states-new}. 
Let $x_0 = \frac{1}{t+1}$ and define the sequence $x_{i+1} = x_i + \frac{1}{2} \sqrt{\frac{x_i (1-x_i)}{t}}$ for $i=0, \dots, m-1$, with $m = \lceil 2 \pi \sqrt{2 t} + 1 \rceil$. According to Lemma~\ref{lem:increasing-distance} (setting $\kappa=1/2$), the sequence reaches the boundary $x_m = 1$.

For each step, Lemma~\ref{lem:condition-intersection} ensures that $\mathcal{G}_{x_i} \cap \mathcal{G}_{x_{i+1}} \neq \emptyset$ since the step size $|x_{i+1} - x_i|$ satisfies the required bound. Consequently, Lemma~\ref{lem:closeness-eigenvalue} implies $D(\phi_{x_i}, \phi_{x_{i+1}}) \le 4\sqrt{\epsilon}$. Combining these bounds via the triangle inequality, we obtain
\begin{equation}
\begin{split}
D(\phi_{x_0}, \phi_{x_m}) &\le 4m\sqrt{\epsilon} \\
&\le (8\pi\sqrt{2t} + 8)\sqrt{\epsilon}.
\end{split}
\end{equation}

At the starting point $x_0$, the probability at $k=0$ is 
\begin{equation}
p_{x_0}(0) = (1-\frac{1}{t+1})^{t} \ge 1/e. 
\end{equation}
Since $1/e > 1/4$, the bound on the bad set probability in Eq.~\eqref{eq:prob_bad_set} implies that $0 \notin \mathcal{B}_{x_0}$, hence $0 \in \mathcal{G}_{x_0}$. This yields $D(\phi_{x_0}, \omega_0) \le 2 \sqrt{\epsilon}$. By a symmetric argument at $x_m=1$, where $p_{1}(t)=1$, we have $t \in \mathcal{G}_{x_m}$ and thus $D(\phi_{x_m}, \omega_t) \le 2\sqrt{\epsilon}$.

Combining these bounds, we have
\begin{equation}
\begin{split}
D(\omega_0, \omega_t) &\le D(\omega_0, \phi_{x_0}) + D(\phi_{x_0}, \phi_{x_m}) + D(\phi_{x_m}, \omega_t) \\
&\le (8 \pi\sqrt{2t} + 12) \sqrt{\epsilon}.
\end{split}
\end{equation}
Finally, since $\omega_0=\mathcal{M}(\rho_0^{\otimes t})$ and $\omega_t=\mathcal{M}(\rho_1^{\otimes t})$, the proof is done.

\subsection{Quantum purity amplification}\label{app:proof_purity_amplification}

We would like to show here that the sample complexity lower bound given by the approximate purification theorem is also met by quantum purity amplification algorithms.
We will first prove the following lemma.

\begin{lemma}[Fidelity bound]\label{lemma:fidelity}
Let $\psi = |\psi\rangle\langle\psi|$ be a pure state, and let $\rho$ be a quantum state such that
\begin{equation}
F(\rho,\psi) := \mathrm{Tr}(\rho \psi) = 1-\epsilon_1.
\end{equation}
Define
\begin{equation}
\epsilon_2 := 1-\mathrm{Tr}(\rho^2).
\end{equation}
Then
\begin{equation}
\epsilon_2  \le 2\epsilon_1.
\end{equation}
\end{lemma}

\begin{proof}
Since the maximal eigenvalue of $\rho$
\begin{equation}
\lambda_{\max}(\rho) \ge \langle \psi | \rho | \psi \rangle = 1-\epsilon_1,
\end{equation}
we have
\begin{equation}
\mathrm{Tr}(\rho^2) \ge \lambda_{\max}(\rho)^2 \ge (1-\epsilon_1)^2.
\end{equation}
Therefore,
\begin{equation}
\epsilon_2
=
1-\mathrm{Tr}(\rho^2)
\le
1-(1-\epsilon_1)^2
=
2\epsilon_1-\epsilon_1^2
\le
2\epsilon_1.
\end{equation}
\end{proof}

Knowing this, the approximate purification theorem can be adjusted into
the following statement.

\begin{corollary}[Approximate no-purification result]\label{coro:approx_no-purification_fidelity}
Let $\mathcal{M}$ be a positive trace-preserving map. Suppose that there exists a parameter $0 \le \epsilon < 1$ such that for any $\rho$ supported on a subspace $\mathcal{D}_S\subseteq\mathcal{D}(\mathcal{H}_d)$ with $\mathrm{dim}(\mathcal{D}_S)\ge 2$, there exists a pure state $\psi_\rho$ with
\begin{equation}
    \Tr(\mathcal{M}(\rho^{\otimes t})\psi_\rho) \ge 1-\epsilon.
\end{equation}
Then there exists a state $\sigma_0$ such that, for all $\rho\in\mathcal{D}_S$,
\begin{equation}
    D\bigl(\sigma_0,\mathcal{M}(\rho^{\otimes t})\bigr) \le (8 \pi\sqrt{2t} + 12) \sqrt{2\epsilon}.
\end{equation}
Here $D(\cdot,\cdot)$ denotes the trace distance $D(\omega,\tau)=\tfrac{1}{2}\Vert \omega-\tau\Vert_1$.
\end{corollary}
Therefore, when the output state $\mathcal{M}(\rho^{\otimes t})$ with different input state $\rho$ can have a large distance, the sample complexity lower bound is also $\Omega(\epsilon^{-1})$ related with the infidelity.
This lower bound is naturally met by the quantum purity amplification protocols.

\begin{fact}[Theorem~II.3 in Ref.~\cite{li2024optimal}]\label{fact:state_purification}
Denote $\psi_0=\ketbra{\psi_0}{\psi_0}$ to be the eigenstate of $\rho$ with the largest eigenvalue and $\lambda_0>\lambda_{1}\ge\cdots\ge\lambda_{d-1}$ to be the eigenvalues of $\rho$.
Then there exists a quantum channel that takes $t$ copies of arbitrary $\rho$ as input and outputs a state $\sigma$ with infidelity $\epsilon \coloneqq 1-\bra{\psi_0}\sigma\ket{\psi_0}$.
The sample complexity $t$ is
\begin{equation}\label{eq:complexity_state_purification}
t=\frac{1}{\epsilon}\sum_{i=1}^{d-1}\frac{\lambda_i}{(\lambda_0-\lambda_i)^2}+\mathcal{O}(1).
\end{equation}
\end{fact}
In quantum state purification, since the target output is the pure eigenstate of the input state, when the actual output state is close to the target, the output states for different input states can differ significantly.
Thus, our $\Omega(\epsilon^{-1})$ naturally holds for this task.
Notice that the complexity shown in Eq.~\eqref{eq:complexity_state_purification} can saturate the lower bound for those states with $\sum_{i=1}^{d-1}{\lambda_i}/{(\lambda_0-\lambda_i)^2}$ being some constant.

\subsection{Quantum principal component analysis}\label{app:proof_qpca}
Another important example that saturates our lower bound is the quantum principal component analysis algorithm~\cite{Lloyd2014qpca}, which transforms input states into unitary evolutions.
To demonstrate saturation, we first need to generalize our approximate purification theorem to cover maps that map states to quantum channels.
Similarly, we first make the following observation.

\begin{lemma}[Diamond norm bound]\label{lemma:diamond}
Let $\mathcal{U}$ be a unitary channel, and let $\mathcal{C}$ be a quantum channel such that
\begin{equation}
D_{\diamond}(\mathcal{C},\mathcal{U}) := \frac{1}{2}\|\mathcal{C}-\mathcal{U}\|_{\diamond} = \epsilon_1.
\end{equation}
Let $\psi=\ketbra{\psi}{\psi}$ to be an arbitrary pure state, and define
\begin{equation}
\epsilon_2 := 1-\mathrm{Tr}(\mathcal{C}(\psi)^2).
\end{equation}
Then
\begin{equation}
\epsilon_2 \le 2\epsilon_1.
\end{equation}
\end{lemma}

\begin{proof}
Let
\begin{equation}
\rho := \mathcal{C}(\psi), \qquad \sigma := \mathcal{U}(\psi).
\end{equation}
Since $\psi$ is pure and $\mathcal{U}$ is unitary, $\sigma$ is also a pure state.
According to the definition of the diamond distance, for any input state $\omega$,
\begin{equation}
D\bigl(\mathcal{C}(\omega),\mathcal{U}(\omega)\bigr) \le D_{\diamond}(\mathcal{C},\mathcal{U}).
\end{equation}
In particular, taking $\omega=\psi$, we obtain
\begin{equation}
D(\rho,\sigma) \le \epsilon_1.
\end{equation}
Using the relationship between the fidelity and trace-distance, we have
\begin{equation}
1-\mathrm{Tr}(\rho\sigma) \le D(\rho,\sigma) \le \epsilon_1,
\end{equation}
and hence
\begin{equation}
\mathrm{Tr}(\rho\sigma) \ge 1-\epsilon_1.
\end{equation}
Adopting the conclusion made in Lemma~\ref{lemma:fidelity}, we have
\begin{equation}
\epsilon_2
=
1-\mathrm{Tr}(\mathcal{C}(\psi)^2)
=
1-\mathrm{Tr}(\rho^2)
\le
2\epsilon_1.
\end{equation}
This proves the claim.
\end{proof}

Combining Theorem~\ref{thm:approx_no-purification_distance_between_states-new}, we can derive the following result.

\begin{corollary}[Approximate no-purification in diamond norm]\label{coro:approx_no-purification_diamond}
Let $\mathcal{M}$ be a linear map that maps unit-trace semi-positive matrices to completely-positive trace-preserving maps.
Suppose that there exists a parameter $0 \le \epsilon < 1$ such that for any $\rho$ supported on a subspace $\mathcal{D}_S\subseteq\mathcal{D}(\mathcal{H}_d)$ with $\mathrm{dim}(\mathcal{D}_S)\ge 2$, there exists a unitary channel $\mathcal{U}_\rho$ with
\begin{equation}
\norm{\mathcal{M}(\rho^{\otimes t})-\mathcal{U}_\rho}_{\diamond}\le\epsilon.
\end{equation}
Then, for arbitrary pure state $\psi$, there exists a state $\sigma_{0,\psi}$ such that, for all $\rho\in\mathcal{D}_S$
\begin{equation}\label{eq:approx_no-purification_diamond}
    D\left[\sigma_{0,\psi},\mathcal{M}(\rho^{\otimes t})(\psi)\right] \le (8 \pi\sqrt{2t} + 12) \sqrt{2\epsilon}.
\end{equation}
\end{corollary}

This corollary means that, if one wants to ensure that the states $\mathcal{M}(\rho^{\otimes t})(\psi)$ have large deviations with different $\rho$ and fixed $\psi$, the sample complexity should scale as $\Omega(\epsilon^{-1})$.
This scaling lower bound is satisfied by the quantum principal component analysis algorithm~\cite{Lloyd2014qpca}.

\begin{fact}[Ref.~\cite{Lloyd2014qpca}]
There is a quantum process that takes $t=\mathcal{O}(\tau^2/\epsilon)$ copies of unknown state $\rho$ as the input and outputs a quantum channel $\mathcal{C}_\rho$, which is close to a unitary channel $\mathcal{U}_\rho(\cdot)=e^{-i\rho \tau}\cdot e^{i\rho \tau}$ in diamond distance
\begin{equation}
\norm{\mathcal{C}_\rho-\mathcal{U}_\rho}=\epsilon.
\end{equation}
\end{fact}
It is easy to see why this upper bound saturates the lower bound given in Corollary~\ref{coro:approx_no-purification_diamond}.
When $\epsilon$ is small, like of order $1/t$, then $\mathcal{M}(\rho^{\otimes t})(\psi)\approx\mathcal{U}_\rho(\psi)$ would have large difference with different $\rho$.
Take the single-qubit case as an example, where we can fix the state $\psi$ as $\ketbra{+}{+}$, where $\ket{\pm}=(\ket{0}\pm\ket{1})/\sqrt{2}$.
If we set the evolution time $\tau=\pi$ and the state $\rho$ to be the maximally mixed state, the output state of $\mathcal{U}_\rho(\psi)$ is unchanged and still $\ketbra{+}{+}$.
If we set the state $\rho=\ketbra{1}{1}$, the output state will be $\ketbra{-}{-}$ as the unitary $e^{i\rho t}=e^{i\ketbra{1}{1}\pi}=\ketbra{0}{0}-\ketbra{1}{1}$ introduces a relative phase on $\ket{1}$.
As $D(\ketbra{+}{+},\ketbra{-}{-})=1$, 
given that $\epsilon\ll1$, the left-hand-side of Eq.~\eqref{eq:approx_no-purification_diamond} is lower bounded by some constant.
Therefore, the scaling of $t=\mathcal{O}(\pi^2/\epsilon)=\mathcal{O}(1/\epsilon)$ saturates the lower bounded given by Eq.~\eqref{eq:approx_no-purification_diamond}.

\subsection{Standard quantum limit (proof of Corollary~\ref{coro:sql})}\label{app:proof_sql}

As stated in the corollary, $\hat{f}$ may not be an unbiased estimator of $f(\rho)$, we thus define the expectation value of it to be another function
\begin{equation}
f^\prime(\rho)=\mathbb{E}[\hat{f}].
\end{equation}
By assumption, we define
\begin{equation}
\Delta_f \coloneqq \sup_\rho f(\rho)-\inf_\rho f(\rho).
\end{equation}
Therefore, there exist two state $\rho_+$ and $\rho_-$ satisfying that
\begin{equation}
\abs{f^\prime(\rho_+)-f^\prime(\rho_-)}=\frac{4}{5}\Delta_f.
\end{equation}
We define the worst-case variance to be
\begin{equation}
V_t=\sup_\rho \mathrm{Var}(\hat{f}).
\end{equation}

We will prove this corollary by constructing a measure-and-prepare channel $\mathcal{M}_{\text{mp}}$.
The measure-and-prepare channel takes $t$ copies of the unknown state $\rho$ as the input, estimates the value of the approximation function $f^{\prime}(\rho)$ with $\hat{f}$, 
then produce the single-qubit state vector 
$(\ket{0}+e^{i\alpha\hat{f}}\ket{1})/\sqrt{2}$, where $\alpha={1}/{\Delta_f}$.
According to its definition, $\mathcal{M}_{\text{mp}}$ is a valid quantum channel, i.e., a completely-positive and trace-preserving map satisfying the requirement of Theorem~\ref{thm:approx_no-purification_distance_between_states-new}.
The mathematical definition of $\mathcal{M}_{\text{mp}}$ is
\begin{equation}\label{eq:mp}
\begin{aligned}
\mathcal{M}_{\text{mp}}(\rho^{\otimes t})=&\mathbb{E}\left[\frac{1}{2}(\ket{0}+e^{-i\alpha\hat{f}}\ket{1})(\bra{0}+e^{i\alpha\hat{f}}\bra{1})\right]\\
=&\frac{1}{2}\begin{bmatrix}
    1 & \mathbb{E}[e^{-i\alpha\hat{f}}]\\
    \mathbb{E}[e^{i\alpha \hat{f}}] & 1
\end{bmatrix},
\end{aligned}
\end{equation}
where the expectation is taken over the different outcomes when measuring the $t$ copies of $\rho$.
The uncertainty in $\hat{f}$ would lead to the decrease in the purity as
\begin{equation}
\Tr\left[\mathcal{M}_{\text{mp}}(\rho^{\otimes t})^2\right]=\frac{1}{2}\left(1+\abs{\mathbb{E}[e^{i\alpha\hat{f}}]}^2\right).
\end{equation}
For the next step, 
define $Y\coloneqq \hat{f}-f^{\prime}(\rho)$. We then find $\mathbb{E}[Y]=0$ and $\mathbb{E}[Y^2]=\mathrm{Var}(\hat{f})$.
Assume that $Y\ll 1$, we have
\begin{equation}
\abs{\mathbb{E}[e^{i\alpha\hat{f}}]}=\abs{\mathbb{E}[e^{i\alpha Y}]}\ge \Re\mathbb{E}[e^{i\alpha Y}]\ge 1-\frac{1}{2}\alpha^2\mathbb{E}[Y^2].
\end{equation}
Knowing that 
\begin{equation}
\mathbb{E}[Y^2]=\mathrm{Var}(\hat{f})\le V_t,
\end{equation}
we find that the purity of the output state 
satisfies
\begin{equation}
\Tr\left[\mathcal{M}_{\text{mp}}(\rho^{\otimes t})^2\right]=1-\epsilon\ge 1-\frac{1}{2}\alpha^2V_t.
\end{equation}

Define the target pure state 
\begin{equation}
\begin{aligned}
\psi_\rho&\coloneqq\frac{1}{2}\left(\ket{0}+e^{i\alpha f^{\prime}(\rho)}\ket{1}\right)\left(\bra{0}+e^{i\alpha f^{\prime}(\rho)}\bra{1}\right)\\
&=\frac{1}{2}\begin{bmatrix}
    1 & e^{-i\alpha f^{\prime}(\rho)}\\
    e^{i \alpha f^{\prime}(\rho)} & 1
\end{bmatrix}.
\end{aligned}
\end{equation}
The trace distance between the target pure state and the actual state shown in Eq.~\eqref{eq:mp} is
\begin{equation}
D\left(\mathcal{M}_{\mathrm{mp}}(\rho^{\otimes t}),\psi_\rho\right)=\frac{1}{2}\abs{\mathbb{E}[e^{i\alpha\hat{f}}]-e^{i\alpha f^{\prime}(\rho)}}=\frac{1}{2}\abs{\mathbb{E}[e^{i\alpha Y}-1]}.
\end{equation}
With the relation that $\abs{e^{iu}-1}\le\abs{u}$ and the Cauchy-Schwarz inequality, we can bound the distance as
\begin{equation}
D\left(\mathcal{M}_{\mathrm{mp}}(\rho^{\otimes t}),\psi_\rho\right)\le\frac{1}{2}\mathbb{E}\abs{[e^{i\alpha Y}-1]}\le\frac{1}{2}\mathbb{E}[\abs{\alpha Y}]\le\frac{1}{2}\alpha\sqrt{V_t}.
\end{equation}
As mentioned before, one can always pick two states $\rho_+$ and $\rho_-$ such that $\abs{f^\prime(\rho_+)-f^\prime(\rho_-)}=\frac{4}{5}\Delta_f$ and thus 
\begin{equation}
D(\psi_{\rho_+},\psi_{\rho_-})=\Omega(1).
\end{equation}
According to the triangle inequality, we have 
\begin{equation}\label{eq:distance_lower_bound}
\begin{aligned}
&D\left(\mathcal{M}_{\mathrm{mp}}(\rho_+^{\otimes t}),\mathcal{M}_{\mathrm{mp}}(\rho_-^{\otimes t})\right)\\
\ge&D\left(\psi_{\rho_+},\psi_{\rho_-}\right)-D\left(\mathcal{M}_{\mathrm{mp}}(\rho_+^{\otimes t}),\psi_{\rho_+}\right)-D\left(\mathcal{M}_{\mathrm{mp}}(\rho_-^{\otimes t}),\psi_{\rho_-}\right)\\
\ge&\Omega( 1-\alpha\sqrt{V_t}).
\end{aligned}
\end{equation}

According to the approximate purification theorem shown in Theorem~\ref{thm:approx_no-purification_distance_between_states-new}, there should exist a fixed state $\sigma_0$ such that for all state $\rho$, we have
\begin{equation}
D\left(\sigma_0,\mathcal{M}_{\mathrm{mp}}(\rho^{\otimes t})\right)\le(8\pi\sqrt{2t}+12)\sqrt{\epsilon}.
\end{equation}
In particular, 
\begin{equation}\label{eq:distance_upper_bound}
\begin{aligned}
&D\left(\mathcal{M}_{\mathrm{mp}}(\rho_+^{\otimes t}),\mathcal{M}_{\mathrm{mp}}(\rho_-^{\otimes t})\right)\\
\le&D\left(\mathcal{M}_{\mathrm{mp}}(\rho_+^{\otimes t}),\sigma_0\right)+D\left(\sigma_0,\mathcal{M}_{\mathrm{mp}}(\rho_-^{\otimes t})\right)\\
\le & 2(8\pi\sqrt{2t}+12)\sqrt{\epsilon}\\
\le &(8\pi\sqrt{2t}+12)\alpha\sqrt{2V_t}.
\end{aligned}
\end{equation}
Combining Eq.~\eqref{eq:distance_lower_bound} and Eq.~\eqref{eq:distance_upper_bound}, we can derive the fundamental lower bound on the variance
\begin{equation}
V_t=\sup_\rho\mathrm{Var}(\hat{f})=\Omega(\Delta_f^2/t).
\end{equation}

\section{Approximate pure dilation (proof of Theorem~\ref{thm:no_approx_dilation})}\label{app:no_approx_purification}

Let $[d]\coloneqq\{1,\dots,d\}$ and let $\Delta_d\coloneqq\{p=(p_1,\dots,p_d): p_i\ge 0,\sum_{i=1}^d p_i=1\}$ denote the $d$-dimensional probability simplex. 
Given any distribution $p \in \Delta_d$, define the diagonal (classical) state
\begin{equation}
    \rho(p)\coloneqq \sum_{i=1}^d p(i)\,\ketbra{i}{i}.
\end{equation}
For a string $z^t\in[d]^t$, define the (classical) distribution
\begin{equation}
    q(z^t)\coloneqq \mathrm{Diag}\bigl(\Tr_B\bigl[ \mathcal{M}(\ketbra{z^t}{z^t})\bigr] \bigr).
\end{equation}

\begin{lemma}[Classical strings]\label{lem:distribution_strings}
For any distribution $p \in \Delta_d$,
\begin{equation}
    \mathbb{E}_{z^t\sim p^{\otimes t}}\left[\mathrm{TV}\bigl(q(z^t),p\bigr)\right]\le \sqrt{\epsilon}+\epsilon+\eta,
\end{equation}
where $\mathrm{TV}(\cdot,\cdot)$ denotes the total variation distance.
\end{lemma}

\begin{proof}
Set $\rho\coloneqq\rho(p), \sigma\coloneqq \mathcal{M}(\rho^{\otimes t})$. By Eq.~\eqref{eq:closeness-eigenvalue-mixed-state}, for the dominant eigenvector $\ket{\phi_{\rho}}_{A,B}$ of $\sigma$, we have
\begin{equation}\label{eq:dominant_overlap}
    1 - F\bigl(\sigma,\ket{\phi_\rho}\bigr) = D\bigl(\sigma,\ket{\phi_\rho}\bigr) \le \epsilon.
\end{equation}

Decomposing the input state as a classical mixture $\rho^{\otimes t} = \bE_{z^t \sim p^{\otimes t}} [\ketbra{z^t}]$, the linearity of $\mathcal{M}$ implies $\sigma = \bE_{z^t \sim p^{\otimes t}} [\mathcal{M}(\ketbra{z^t})]$. Since the fidelity with a pure state is linear in its first argument, we have $F(\sigma, \ket{\phi_\rho}) = \bE_{z^t}[F(\mathcal{M}(\ketbra{z^t}), \ket{\phi_\rho})]$. Substituting this into Eq.~\eqref{eq:dominant_overlap} yields
\begin{equation}
    \mathbb{E}_{z^t}\bigl[1 - F\bigl(\mathcal{M}(\ketbra{z^t}{z^t}),\,\ket{\phi_\rho}\bigr)\bigr] \le \epsilon.
\end{equation}
For each $z^t$, the Fuchs–van de Graaf inequality with a 
pure state implies
\begin{equation}
    D\bigl(\mathcal{M}(\ketbra{z^t}{z^t}),\,\ket{\phi_\rho}\bigr) \le \sqrt{1 - F\bigl(\mathcal{M}(\ketbra{z^t}{z^t}),\,\ket{\phi_\rho}\bigr)}.
\end{equation}
By Jensen’s inequality, we get
\begin{equation}
\begin{aligned}
    &\mathbb{E}_{z^t}\Bigl[D\bigl(\mathcal{M}(\ketbra{z^t}{z^t}),\,\ket{\phi_\rho}\bigr)\Bigr]\\
    \le& \mathbb{E}_{z^t}\Bigl[\sqrt{1 - F\bigl(\mathcal{M}(\ketbra{z^t}{z^t}),\,\ket{\phi_\rho}\bigr)}\Bigr]\\
    \le& \sqrt{\mathbb{E}_{z^t}\bigl[1 - F\bigl(\mathcal{M}(\ketbra{z^t}{z^t}),\,\ket{\phi_\rho}\bigr)\bigr]}\\
    \le& \sqrt{\epsilon}.
\end{aligned}
\end{equation}
Define $\phi_{p,A}\coloneqq \mathrm{Diag}\bigl(\Tr_B(\ketbra{\phi_\rho}{\phi_\rho})\bigr)$. By data processing (partial trace followed by dephasing),
\begin{equation}
    \mathbb{E}_{z^t}\bigl[\mathrm{TV}\bigl(q(z^t),\phi_{p,A}\bigr)\bigr] \le \sqrt{\epsilon}.
\end{equation}
Let $\bar q \coloneqq \mathbb{E}_{z^t} q(z^t) = \mathrm{Diag}\bigl(\Tr_B(\sigma)\bigr)$. By the near-purity condition and data processing,
\begin{equation}
    \mathrm{TV}\bigl(\bar q,\phi_{p,A}\bigr) \le \epsilon.
\end{equation}
By the marginal-accuracy condition and data processing,
\begin{equation}
    \mathrm{TV}\bigl(\bar q,p\bigr) \le \eta.
\end{equation}
Finally, by the triangle inequality,
\begin{equation}
\begin{aligned}
    &\mathbb{E}_{z^t}\bigl[\mathrm{TV}\bigl(q(z^t),p\bigr)\bigr]\\
    \le &\mathbb{E}_{z^t}\bigl[\mathrm{TV}\bigl(q(z^t),\phi_{p,A}\bigr)\bigr] + \mathrm{TV}\bigl(\phi_{p,A},\bar q\bigr) + \mathrm{TV}\bigl(\bar q,p\bigr)\\
    \le& \sqrt{\epsilon} + \epsilon + \eta.
\end{aligned}
\end{equation}
\end{proof}

We now view $q$ as a function $q:[d]^t\to\Delta_d$. Intuitively, when $t$ is small, any function based on $t$ samples drawn from $p$ cannot produce a good estimator of the underlying distribution $p$. This is a well-studied problem in distribution estimation, and we use the following bound~{\cite[Theorem 1]{Han2015Minimax}}:

\begin{lemma}[Minimax estimation of distributions]\label{lem:minimax_estimation}
    For any function $f:[d]^t\to\Delta_d$ and $t > \frac{ed}{32}$, 
    \begin{equation}
    \begin{aligned}
        &\sup_{p \in \Delta_d} \mathbb{E}_{z^t \sim p^{\otimes t}} \norm{f(z^t) - p}_1 \\
        \ge & \frac{1}{8}\sqrt{\frac{ed}{2t}} - \exp(-\frac{t}{24}) - 12 \exp(-\frac{d}{32 (\ln d)^2}).
    \end{aligned}
    \end{equation}
\end{lemma}

\begin{proof}[Proof of Theorem \ref{thm:no_approx_dilation}]
Lemma \ref{lem:distribution_strings} tells us that
\begin{equation}
    \delta_t \coloneqq \sup_{p\in\Delta_d}\, \mathbb{E}_{z^t\sim p^{\otimes t}} \norm{q(z^t)-p}_1 \le2( \sqrt{\epsilon}+\epsilon+\eta) \eqqcolon \xi.
\end{equation}
Lemma \ref{lem:minimax_estimation} tells us that for $t > \frac{ed}{32} $,  
\begin{equation}
\begin{split}
    \delta_{t} &\ge \frac{1}{8}\sqrt{\frac{ed}{2t}} - \exp\!\left(-\frac{t}{24}\right) - 12\exp\!\left(-\frac{d}{32(\ln d)^2}\right) \\
    &\ge \frac{1}{8}\sqrt{\frac{ed}{2t}} - 13\exp\!\left(-\frac{d}{32(\ln d)^2}\right),
\end{split}
\end{equation}
where the last equation holds for $d \ge 20$. 
Combining with $\delta_t\le \xi$, we obtain
\begin{equation}
    \xi + 13\exp\left(-\frac{d}{32(\ln d)^2}\right) \ge \frac{1}{8}\sqrt{\frac{ed}{2t}},  
\end{equation}
which gives the desired lower bound
\begin{equation}
    t \ge \frac{ed}{128 \left[\xi + 13\exp\left(-\frac{d}{32(\ln d)^2}\right)\right]^2}.
\end{equation}
Next, note that $\delta_t \ge \delta_m$ for $t\le m$. Hence, for any $t\le \frac{ed}{32}$,
\begin{equation}
    \delta_t \ge \delta_{ed/32 + 1} \ge  \frac{1}{2} - o(1).
\end{equation}
Thus, for sufficiently large $d$ and any $\xi \le \frac{1}{3}$, one cannot have $\delta_t \le \xi$ when $t\le \frac{ed}{32}$. This completes the proof.
\end{proof}

\section{Gaussian purifications}\label{app:Gaussian}
\subsection{Preliminaries of Gaussian purification}

In this section, we discuss notions of purification for bosonic Gaussian quantum states. Given the central role of such states and operations, this focus is well motivated. 
They feature strongly in optical implementations -- after all, light is captured by bosonic modes of this kind -- but also in 
many other physical settings, such as those in which motional degrees of freedom play a role.
Purifications are particularly important in the study of quantum thermodynamics, where they underpin many structural results. Moreover, they have featured prominently in recent work on random purifications, and for good reason.
Here, we present the proofs of the statements given in the main text. As before, we consider 
$t$ copies of quantum systems, each consisting of 
$k$ bosonic modes. Altogether, the input therefore comprises $kt$ modes. Such systems are captured by $2kt$ canonical operators $(x_1,p_1, \dots, x_{kt}, p_{kt})$ featuring canonical commutation relations that can be captured by the
skew-symmetric symplectic matrix
\begin{equation}
	\sigma\coloneq \bigoplus_{j=1}^{kt} \left[
	\begin{array}{cc}
		0 & 1\\
		-1 & 0 \\
	\end{array}
	\right].
\end{equation}	
We consider also symplectic matrix of $k$ output modes, again denoted by $\sigma$.
For such systems, we consider Gaussian quantum states, which are by far the most important class of states. They are completely characterized by first moments, which 
are not relevant here and will without loss of generality be set to zero, and second moments that can be collected in a real symmetric \emph{covariance matrix}
$\Gamma \in \mathbb{R}^{2kt\times 2kt}$. The specific input of the form $\rho^{\otimes t}$ with $\rho$ being a Gaussian quantum state is on the  level of second moment characterized by a direct sum
$\gamma^{\oplus t}$, 
so a block diagonal matrix of $2k\times 2k$ blocks each.

We now turn to the completely
positive maps ${\cal M}$ considered here. They are naturally \emph{Gaussian channels}, so 
completely positive trace-preserving maps that map Gaussian quantum states onto Gaussian states. 
For the specific case at hand,  taking systems of $kt$ modes as input and having individual systems of $k$ modes as output,
they act on covariance matrices as
\begin{equation}
\gamma\mapsto \Gamma = A\gamma^{\oplus t} A^T + B,
\end{equation}
where $B \in \mathbb{R}^{2k\times 2k}$ is real and symmetric, while
$A \in \mathbb{R}^{2k\times 2kt}$ is real and has no further structure.
Complete positivity requires that
\begin{equation}\label{CP}
B+ i\sigma - i A \sigma A^T \geq 0.
\end{equation}
Important special cases are Gaussian unitary channels, for which $B=0$ and $A$ reflects a symplectic transformation, 
which means that (for suitable dimensions of the symplectic form) $A\sigma A^T = \sigma$. Important are also
\emph{passive Gaussian channels}, which are realized as a dilation involving passive transformations, which are not only symplectic transformations but also orthogonal ones and where the environment is initially prepared in the vacuum. Such passive transformations are of key importance, as they reflect the non-squeezing channels and can be seen as arising from linear optical networks. For those, $\|A\|\leq 1$ and $B\leq 1$.

\subsection{Proof of Observation \ref{obs:no-Gaussian-purification}: Exact Gaussian purification}

We are now in the position to prove the actual observation. 
We first show that having an exactly pure output requires replacing the output by a fixed Gaussian
quantum state. Let $\gamma$ be the covariance matrix of a strictly mixed quantum state $\rho$ of $k$ modes. The output of ${\cal M}$ has the covariance matrix
\begin{equation}\label{Out}
\Gamma = A\gamma^{\oplus t} A^T + B. 
\end{equation}
This reflects a pure state exactly if all 
symplectic eigenvalues of $\Gamma$ are $1$,
which means that the standard positive eigenvalues of the matrix $i \sigma \Gamma$, the \emph{symplectic eigenvalues}, 
are given by $1$. Now we can conclude from complete positivity captured in Eq.\ (\ref{CP}) 
 that $B \geq -i\sigma+ i A \sigma A^T$, which combined with Eq.\ (\ref{Out}) gives
\begin{equation}
\Gamma + i \sigma \geq A(\gamma^{\oplus t} + i\sigma) A^T .
\end{equation}
A defining property of covariance matrices is that $\gamma^{\oplus t} + i\sigma\geq 0$: This is actually nothing but the mathematically 
precise formulation of the
Heisenberg uncertainty principle for the given family input states. Since $\Gamma$ is the covariance matrix of a pure Gaussian state,
there exists a symplectic transformation $S$ satisfying $S\sigma S^T = \sigma$ so that 
\begin{equation}
\mathbb{I} + i \sigma \geq S A(\gamma^{\oplus t} + i\sigma) A^T S^T  \geq 0
\end{equation}
must be true. The expression in the left hand side is also maximally rank-deficient, in that half of the eigenvalues are zero.
This means that half the eigenvalues of 
$S A(\gamma^{\oplus t} + i\sigma) A^TS^T $ are also zero. Since $\gamma^{\oplus t} + i\sigma$ has no 
symplectic eigenvalue $1$, this means that $SA(\gamma^{\oplus t} + i\sigma) A^TS^T = A(\gamma^{\oplus t} + i\sigma) A^T=0$. 
But this means that $B$ is the covariance matrix of a pure Gaussian state, 
so that the symplectic eigenvalues of $B$ are all $1$. This  again means that
\begin{equation}
{\cal M}(\rho^{\otimes t}) =\Psi,
\end{equation}
with a pure Gaussian state $\Psi$, 
independent of $\rho$.

\subsection{Proof of Theorem \ref{thm:approx_no-purification_distance_between_states-Gaussian}: Approximate Gaussian purification}

We now turn to proving the statement on approximate Gaussian purification.
We take as the covariance matrix of the reference Gaussian quantum state $\Psi$ the matrix
\begin{equation}
\xi:= B+ A A^T .
\end{equation}
This is clearly a real symmetric matrix, and the Heisenberg uncertainty
principle $\xi+ i\sigma$ is satisfied because of
\begin{eqnarray}
B+ A A^T + i\sigma &\geq& 
B+ A A^T + i\sigma -  A (\mathbb{I}+ i\sigma ) A^T \\
&\geq& 0.\nonumber
\end{eqnarray}
The difference between the output covariance matrix $\Gamma$ and the reference one $\xi$
is given by
\begin{equation}
A \gamma^{\oplus t} A^T +B - (B +A A^T) = A(\gamma^{\oplus t} -\mathbb{I})A^T=: D.
\end{equation}
which will be made use of soon. This matrix is positive semi-definite, by assumption that 
$\gamma\geq \mathbb{I}$.
We now turn to discussing the
purity of the output state $\omega= {\cal M}(\rho^{\otimes t})$, which is given by ${\rm tr}(\omega^2)$.
Making use of a standard formula for the
purity expressed in terms of covariance matrices, one finds that 
the condition of a large purity ${\rm tr}(\omega^2)\geq 1-\epsilon$ 
implies that
\begin{equation}
{\rm tr}(\omega^2) = \frac{1}{\det(\Gamma)^{1/2}}\geq 1-\epsilon.
\end{equation}
From this, it follows that
\begin{equation}
\det(\Gamma)^{1/2} \leq 1+ 2\epsilon.
\end{equation}
Since this is nothing but
\begin{equation}
\det(A \gamma^{\oplus t} A^T +B)^{1/2} \leq 1+ 2\epsilon,
\end{equation}
and hence one finds
\begin{equation}
\det(A (\gamma^{\oplus t} -\mathbb{I}) A^T +\xi)^{1/2} \leq 1+ 2\epsilon,
\end{equation}
which is
\begin{equation}
\det(D +\xi) \leq (1+ 2\epsilon)^2.
\end{equation}
To progress, note that
\begin{equation}
\tr(D +\xi) \leq  \tr(\xi)-\|\xi\| + \frac{(1+ 2\epsilon)^2 \| \xi\|}{\det(\xi)}.
\end{equation}
Subtracting $\tr(\xi) $ yields
\begin{equation}
\tr(D ) \leq   \frac{(1+ 2\epsilon)^2 \| \xi\|}{\det(\xi)}-\|\xi\|.
\end{equation}
This gives
\begin{equation}
\|D \|_1 \leq  \| \xi\|
\left(
\frac{(1+ 2\epsilon)^2}{\det(\xi)}-1
\right).
\end{equation}
Now, since $\det(\xi)\geq 1$, as is true for every covariance matrix of a quantum state,
\begin{equation}
\|D \|_1 \leq  8\epsilon \| \xi\| \leq 16
 \epsilon ,
 \end{equation}
as $0\leq \epsilon< 1$. For the final bound, we also need to upper bound the operator norm of the sum of the 
output and the given covariance matrix,
given by
\begin{equation}
\| \xi+\Gamma  \| =  \| \xi +A\gamma^{\oplus t} A^T +   B\|=  \| 2 \xi +A(\gamma^{\oplus t} -\mathbb{I}) A^T  \|.
 \end{equation}
The triangle 
inequality gives
\begin{equation}
\| \xi+\Gamma  \| \leq 2  \|  \xi \| +\| D  \|.
 \end{equation}
 For the passive Gaussian channel at hand,  $\|B\|\leq 1$. And $\|A A^T\|\leq 1$. This implies that  $\|  \xi \|\leq 2$, which implies that 
 \begin{equation}
\| \xi+\Gamma  \| \leq 4 +\| D  \| \leq 4+ 16\epsilon.
 \end{equation}
 Then, finally, invoking the results of Ref.\ \cite{Bittel2025optimalestimatesof}, relating closeness in trace distance 
 to closeness in second moments, we find
 \begin{eqnarray}
 	D(\omega, \Psi) &\leq &
	\frac{1+\sqrt{3} }{16}\|D\|_1 \, \|	\xi+\Gamma  \|\nonumber\\
	&\leq &4(1+\sqrt{3}) \epsilon (1+4 \epsilon).
 \end{eqnarray}


%

\end{document}